\documentclass[onecolumn,draftcls]{IEEEtran}

\usepackage{graphicx}
\usepackage{bm}
\usepackage{url}
\usepackage{amsfonts,amssymb,amsmath}
\usepackage{mystyle}
\allowdisplaybreaks                              
\centerfigcaptionstrue
\usepackage{cite}

\def\dyn{{d}}
\def\stat{{s}}

\def\SNRaEff{\rho_\dyn}
\def\SNRbEff{\rho_\stat}

\def\Expt{\mathbb E}
 \def\Tr{\text{tr}}
\def\Nn{N_\dyn}
\def\Nc{N_\stat}
\def\NTx{M}
\def\tr{\text{tr}}
\def\Heqd{\widetilde{\bH}_\dyn}
\def\Heqs{\widetilde{\bH}_\stat}
\def\vec{\mathbf {vec}}

\def\defeq{\stackrel{\triangle}{=}}

\begin{document}

\title{Coherent Product Superposition for Downlink Multiuser MIMO }
\author{Yang Li, {\em Student Member, IEEE}, and Aria Nosratinia, {\em
    Fellow, IEEE}
  \thanks{The authors are with the Department of
    Electrical Engineering, University of Texas at Dallas, Richardson,
  TX 75080, USA, email: aria@utdallas.edu,
  yang@utdallas.edu}
}

\maketitle

\begin{abstract}
In a two-user broadcast channel where one user has full CSIR and the
other has none, a recent result showed that TDMA is strictly
suboptimal and a product superposition requiring non-coherent
signaling achieves DoF gains under many antenna configurations.
This work introduces product superposition in the domain of coherent
signaling with pilots, demonstrates the advantages of product
superposition in low-SNR as well as high-SNR, and establishes DoF
gains in a wider set of receiver antenna configurations. Two
classes of decoders, with and without interference cancellation, are
studied. Achievable rates are established by analysis and illustrated
by simulations.
\end{abstract}

\begin{keywords}
CSIR, superposition, degrees of freedom, pilot, channel estimation
\end{keywords}

\section{Introduction}
\label{sec:Introduction}

Due to varying mobility and the effects of the propagation
environment, wireless network nodes often have unequal capability to
acquire CSIR (channel state information at receiver). Downlink
(broadcast) transmission to nodes with unequal CSIR is therefore a
subject of practical interest.

It has been known that if all downlink users have full CSIR, then
orthogonal transmission (e.g. TDMA) achieves the optimal degrees of
freedom (DoF)~\cite{Caire2003, Huang2009}, in the absence of CSIT under fast fading. A
similar result is known to hold for certain antenna configurations in
the absence of CSIR. Recently it was discovered~\cite{Li-TIT2012-2}
that a very different behavior emerges when one user has perfect CSIR
and the other has none: in this case TDMA is highly suboptimal and a
product superposition can achieve gains in the degrees of freedom
(DoF). However, this result~\cite{Li-TIT2012-2} required non-coherent
Grassmannian signaling while most practical systems use pilots and
employ coherent detection after channel estimation. In addition, the
result~\cite{Li-TIT2012-2} was limited to high-SNR and did not
demonstrate optimality in all receiver antenna configurations.

In this paper we extend the product superposition to coherent
signaling with pilots. 
This is motivated by several factors, among them the popularity and
prevalence of coherent signaling in the practice of wireless
communications, as well as the known results in the point-to-point
channel~\cite{Zheng2002} showing that pilot-based transmission can
perform almost as well as Grassmannian signaling. We show that a
similar result holds in the mixed-mobility broadcast channel.  
In the process, we demonstrate the DoF gains of product
superposition for more antenna configurations than
in~\cite{Li-TIT2012-2}, and in addition show that it has excellent
performance in low-SNR as well as high-SNR.

A downlink scenario with two users is considered in this paper, where
one user has a short coherence interval and is referred to as the {\em
  dynamic user}, and the other has a long coherence interval and is
referred to as the {\em static user}.  The main results of this paper
are as follows.

\begin{itemize}
\item 

We propose a new signaling structure that is a product of two matrices
representing the signals of the static and dynamic user,
respectively, where the data for both users are transmitted using coherent
signaling.

\item
We propose two decoding methods.  The first method performs no
interference cancellation at the receiver.  We show that under this
method, at both high SNR and low SNR, the dynamic user experiences
almost no degradation due to the transmission of the static
user. Therefore in the sense of the cost to the other user, the static
user's rate is added to the system ``for free.''  Avoiding
interference cancellation gives this method the advantage of
simplicity.

\item 
The second method further improves the static user's rate by allowing
it to decode and remove the dynamic user's signal. This increases the
effective SNR for the static user and provides further rate gain.

\item
We show that the product superposition has DoF gains when the dynamic
user has either more, less or equal number of antennas as the static
user. Previously~\cite{Li-TIT2012-2} the DoF gain was
demonstrated only when the dynamic user had fewer or equal number of
antennas compared with the static user.

\end{itemize}

The following notation is used throughout the paper: for a matrix
$\bA$, the transpose is denoted with $\bA^t$, the conjugate transpose
with $\bA^H$, the pseudo inverse with $\bA^{\dag}$ and the element in
row $i$ and column $j$ with $[\bA]_{ij}$.  The $k\times k$ identity
matrix is denoted with $\bI_k$. The set of $n\times m$ complex
matrices is denoted with $\mathcal{C}^{n\times m}$. We denote
$\mathcal{CN}(0,1)$ as the circularly symmetric complex Gaussian
distribution with zero mean and unit variance. For all variables
the subscripts ``s'' and ``d'' stand as mnemonics for ``static'' and
``dynamic'', respectively, and subscripts ``$\tau$'' and ``$\delta$''
stand for ``training'' and ``data.''

\section{System Model and Preliminaries}
\label{sec:Preliminaries}

We consider an $\NTx$-antenna base-station transmitting to two users,
where the dynamic user has $\Nn$ antennas and the static user has
$\Nc$ antennas. The channel coefficient matrices of the two users are
$\bH_\dyn\in\mathcal{C}^{\Nn\times \NTx}$ and
$\bH_\stat\in\mathcal{C}^{\Nc\times \NTx}$, respectively. In this
paper we restrict our attention to $\NTx=\max\{\Nn,\Nc\}$. The system
operates under block-fading, where $\bH_\dyn$ and $\bH_\stat$ remain
constant for $T_\dyn$ and $T_\stat$ symbols, respectively, and change
independently across blocks. The coherence time $T_\dyn$ is small but
$T_\stat$ is large ($T_\stat\gg T_\dyn$) due to different
mobilities. The difference in coherence times means that the channel
resources required by the static user to estimate its channel are
negligible compared to the training requirements of the dynamic
user. To reflect this in the model, it is assumed that $\bH_\stat$ is
known by the static user (but unknown by the dynamic user, naturally),
while $\bH_\dyn$ is not known {\em a priori} by either user.

\begin{figure}
\centering \includegraphics[width=2.5in]{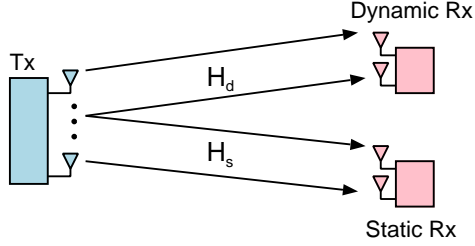}
\caption{Channel model.}
\label{fig:Channel_Model}
\end{figure}

Over $T_\dyn$ time-slots (symbols) the base-station sends
$\bX=[\mathbf{x}_1,\cdots,\mathbf{x}_M]^t$ across $\NTx$ antennas,
where $\mathbf{x}_i\in \mathcal{C}^{T_\dyn \times 1}$ is the signal vector
sent by the antenna $i$. The signal at the dynamic and static users is
respectively
\begin{align}
\bY_\dyn &= \bH_\dyn \bX + \bW_\dyn, \nonumber \\
\bY_\stat & = \bH_\stat \bX + \bW_\stat, \label{eq:channelmodel}
\end{align}
where $\bW_\dyn\in \mathcal{C}^{\Nn \times T_\dyn}$ and $\bW_\stat\in
\mathcal{C}^{\Nc \times T_\dyn}$ are additive noise with i.i.d. entries
$\mathcal{CN}(0,1)$. Each row of $\bY_\dyn\in \mathcal{C}^{\Nn \times T_\dyn}$
(or $\bY_\stat\in \mathcal{C}^{\Nc \times T_\dyn}$) corresponds to the received
signal at an antenna of the dynamic user (or the static user) over $T_\dyn$
time-slots.  The base-station is assumed to have an average power
constraint $\rho$
\begin{equation}
\Expt\big[\sum_{i=1}^{\NTx}\tr(\mathbf{x}_i\mathbf{x}_i^{H})\big] =
\rho \,T_\dyn. \label{eq:powerconstraint}
\end{equation}


The channels $\bH_\dyn$ and $\bH_\stat$ have
i.i.d. entries with the distribution $\mathcal{CN}(0,1)$. We assume
$\NTx = \max(\Nn,\Nc)$ and $T_\dyn\ge 2\Nn$~\cite{Zheng2002}.

\subsection{The Baseline Scheme}
\label{sec:baseline}

We start by establishing a baseline scheme and outlining its capacity
for the purposes of comparison. In our system model, MIMO transmission
schemes involving dirty paper coding, zero-forcing, or similar
techniques~\cite{Weingarten2006,Jindal2006,Yoo2006,Sharif2005} are not
applicable since $\bH_\dyn$ varies too quickly for feedback to
transmitter. Our baseline method uses orthogonal
transmission, i.e., TDMA.

For the dynamic user, we consider the following near-optimal
method. The base-station activates only $\Nn$ out of $\NTx$
antennas~\cite{Zheng2002}, sends an orthogonal pilot matrix
$\bS_{\tau}\in\mathcal{C}^{\Nn\times\Nn}$ during the first $\Nn$
time-slots, and then sends
i.i.d. $\mathcal{CN}(0,1)$ data signal $\bS_{\delta}\in
\mathcal{C}^{\Nn\times (T_\dyn-\Nn)}$ in the following $T_\dyn-\Nn$
time-slots~\cite{Hassibi2003}, that is
\begin{equation}
  \bX = \bigg[\sqrt{\frac{\rho_{\tau}}{\Nn} } \,\bS_{\tau} \;
    \sqrt{\frac{\rho_{\delta}}{\Nn}}\,\bS_{\delta}\bigg]
\end{equation}
where 
$\bS_{\tau}\bS_{\tau}^H = \Nn\bI$, and $\rho_{\tau}$ and
$\rho_{\delta}$ are the average power used for training and data,
respectively, and satisfy the power constraint
in~\eqref{eq:powerconstraint}:
\begin{equation}
\rho_{\tau} \Nn + \rho_{\delta}(T_\dyn-\Nn) \le \rho T_\dyn. \label{eq:powerconstraint1}
\end{equation}
The dynamic user employs a linear minimum-mean-square-error (MMSE)
estimation on the channel. The normalized channel estimate obtained in
this orthogonal scheme is denoted
$\overline{\bH}_\dyn\in\mathcal{C}^{\Nn\times \Nn}$.
Under this condition, the rate attained
by the dynamic user is~\cite{Hassibi2003}:
\begin{equation}
R_\dyn \ge (1-\frac{\Nn}{T_\dyn}) \Expt \big[\log\det(\bI_{\Nn} +
  \frac{\SNRaEff}{\Nn} \overline{\bH}_\dyn \overline{\bH}_\dyn^{H})\big],  \label{eq:R10}
\end{equation}
where $\SNRaEff$ is the effective signal-to-noise ratio (SNR)
\begin{equation}
\SNRaEff = \frac{\rho_{\delta}\,\rho_{\tau}}{1+\rho_{\delta} + \rho_{\tau}\Nn}.\label{eq:rho10}
\end{equation}

For the static user, the channel is known at the receiver, the
base-station sends data directly using all $M$ antennas. The rate
achieved by the static user is~\cite{Telatar1999}
\begin{equation}
R_\stat = \Expt \bigg[\log\det\bigg(\bI_{\Nc} +
  \frac{\rho}{\Nc} \, \bH_\stat \bH_\stat^{H}\bigg)\bigg]. \label{eq:R20}
\end{equation}

Time-sharing ($0\le p\le 1$) between $R_\dyn$ and $R_\stat$ yields the rate region
\begin{equation}
\mathcal{R}_{OT} = \big(pR_\dyn,\, (1-p)R_\stat \big). \label{eq:TDMA}
\end{equation}

\subsection{Overview of Product Superposition~\protect\cite{Li-TIT2012-2} }
\label{sec:overview}

In~\cite{Li-TIT2012-2}, a product superposition based on Grassmannian
signaling was proposed and  shown to achieve significant gain in DoF
over orthogonal transmission. 
%
In the so-called 
{\em Grassmannian-Euclidean superposition}~\cite{Li-TIT2012-2}, the base-station transmits
\begin{equation}
\bX=\bX_\stat\bX_\dyn \in\mathcal{C}^{M\times T_\dyn}
\end{equation}
over $T_\dyn$ time-slots, where $\bX_\dyn\in\mathcal{C}^{\Nn\times T_\dyn}$ and
$\bX_\stat\in\mathcal{C}^{M\times\Nn}$ are the signals for the dynamic
and static user, respectively. 
For the dynamic user, a Grassmannian (unitary) signal is used to
construct $\bX_\dyn$, so that information is carried only in the
subspace spanned by the rows of  $\bX_\dyn$.
As long as  $\bX_\stat$ is full rank, its multiplication does not create
interference for the dynamic user, since $\bX_\stat\bX_\dyn$ and
$\bX_\dyn$ span the same row-space. 

The static user decodes and peels off $\bX_\dyn$ from the received
signal, then decodes $\bX_\stat$, which carries information in the usual
manner of space-time codes.

In conventional point-to-point non-coherent methods~\cite{Zheng2002, Brehler2001},
power gain is obtained at low-SNR and yet no DoF gain is achieved. Compared with 
these method, the product superposition attains DoF gain by transmitting to 
two users.
\section{Pilot-Based Product Superposition}
\label{sec:PBPS}

We now develop a product superposition with coherent signaling for the
two-user broadcast channel.  We start with a simple method with
single-user decoding (no interference cancellation).

\subsection{Signaling Structure}

Over $T_\dyn$ symbols (the coherence interval of the dynamic
user) the base-station sends $\bX\in\mathcal{C}^{M\times T_\dyn}$
across $\Nc$ antennas:
\begin{equation}
\bX= \bX_\stat \bX_\dyn, \label{eq:Tx1}
\end{equation}
where $\bX_\stat\in \mathcal{C}^{M\times \Nn}$ is the data matrix for
the static user and has i.i.d. $\mathcal{CN}(0,1)$ entries. The signal
matrix $\bX_\dyn \in \mathcal{C}^{\Nn\times T_\dyn}$ is intended for the dynamic user
and consists of the data matrix $\bX_{\delta} \in \mathcal{C}^{\Nn\times
  (T_\dyn-\Nc)}$ whose entries are i.i.d. $\mathcal{CN}(0,1)$ and the
pilot matrix $\bX_{\tau} \in \mathcal{C}^{\Nn\times \Nc}$ which is
{\em unitary}, and is known to both static and dynamic users. 
\begin{equation}
\bX_\dyn = \bigg[ \sqrt{c_{\tau}}\;\bX_{\tau}\; \sqrt{c_{\delta}}\;\bX_{\delta}\bigg], \label{eq:X1}
\end{equation}
where the constant $c_{\tau}$ and $c_{\delta}$ satisfy the power
constraint~\eqref{eq:powerconstraint}:
\begin{equation}
\Nc\Nn \big( c_{\tau} + (T_\dyn-\Nn)c_{\delta} \big) \le \rho
\,T_\dyn \label{eq:powerconstraint2}.
\end{equation}

Please make note of the normalization of pilot and data
matrices in the product superposition: The pilot matrix is unitary,
i.e., the entire pilot power is normalized, while the data matrix is
normalized per time per antenna. This is only for convenience of
mathematical expressions in the sequel; full generality is maintained
via multiplicative constants $c_{\delta}$ and $c_\tau$.

A sketch of the ideas involved in the decoding at the dynamic and static
users is as follows. The
signal received at the dynamic user is
\begin{equation}
\bY_\dyn = \bH_\dyn\bX_\stat \bigg[\sqrt{c_{\tau} }\bX_{\tau} \; \sqrt{c_{\delta}}\bX_{\delta}\bigg] +
\bW_\dyn
\end{equation}
where $\bW_\dyn$ is the additive noise. The dynamic
user uses the pilot matrix to estimate the equivalent channel
$\bH_\dyn\bX_\stat$, and then decodes $\bX_{\delta}$ based on the channel
estimate.

For the static user, the signal received during the first $\Nn$
time-slots is
\begin{equation}
\bY_{\stat 1 } = \sqrt{c_{\tau} }\;\bH_\stat \bX_\stat \bX_{\tau} +
\bW_{\stat 1}
\end{equation}
where $\bW_{\stat 1}$ is the additive noise at the static user during
the first $\Nn$ samples. The static user multiplies its received
signal by $\bX_{\tau}^H$ from the right and then recovers 
\footnote{The rate is assumed to be smaller than the channel capacity, so the codeword (multiple blocks of $\bX_\stat$) can be always decoded as long as it is sufficient long.}
the signal
$\bX_\stat$.

\begin{remark}
Each of the dynamic user's codewords includes pilots because it needs frequent
channel estimates. No pilots are included in the individual codewords
of the static user because it only needs infrequent channel estimate
updates. In practice static user's channel training occurs at much longer
intervals outside the proposed signaling structure.
\end{remark}

\subsection{Main Result}

\begin{theorem}
\label{thm:rate1}
Consider an $M$-antenna base-station, 
a dynamic user with $\Nn$-antennas and coherence time $T_\dyn$, and a
static user with $\Nc$-antennas and coherence time $T_\stat\gg
T_\dyn$. Assuming the dynamic user does not know its channel $\bH_\dyn$ but
the static user knows its channel $\bH_\stat$, the pilot-based product
superposition achieves the rates
\begin{align}
R_\dyn & =  (1-\frac{\Nn}{T_\dyn}) \Expt \bigg[\log\det\bigg(\bI_{\Nn} +
  \frac{\SNRaEff}{\Nn} \overline{\bH}_\dyn\overline{\bH}_\dyn^{H}\bigg)\bigg],  \label{eq:DynamicRate} \\
R_\stat & = \frac{\Nn}{T_\dyn} \,\Expt \bigg[\log\det\bigg(\bI_{\Nc} +
  \frac{\SNRbEff}{\Nc} \, \bH_\stat \bH_\stat^{H}\bigg)\bigg], \label{eq:StaticRate}
\end{align}
where $\overline{\bH}_\dyn$ is the {\em normalized} MMSE channel estimate of
the equivalent dynamic channel
$\bH_\dyn\bX_\stat$, and $\SNRaEff$ and $\SNRbEff$ are the effective
SNRs:
\begin{align}
\SNRaEff & =   \frac{c_{\tau}c_{\delta}\Nn\Nc^2}{1+c_{\tau}\Nc +  c_{\delta}\Nn\Nc},\\
\SNRbEff & = c_{\tau}\Nc.
\label{Eq:rho2}
\end{align}
\end{theorem}
\begin{proof}
See Appendix~\ref{apdx:thm1}.
\end{proof}

For the static user, the effective SNR $\SNRbEff$
increases linearly with the power used in the training of the dynamic
user. This is because the static user decodes based on the signal
received during the training phase of the dynamic user.

For the dynamic user, the effective SNR $\SNRaEff$ is unaffected
by superimposing $\bX_\stat$ on $\bX_\dyn$. To see this,
compare~\eqref{eq:powerconstraint1}
with~\eqref{eq:powerconstraint2} to arrive at $\rho_{\tau}=c_{\tau}\Nc$
and $\rho_{\delta}=c_{\delta}\Nn\Nc$,  therefore the two SNRs are equal to
\begin{equation}
\SNRaEff = \frac{c_{\tau}c_{\delta} \Nn\Nc^2}{1+c_{\tau}\Nc +
  c_{\delta}\Nn\Nc}.  
\label{eq:SNRdynamic}
\end{equation}

Intuitively, the rate available to the dynamic user via orthogonal
transmission (Eq.~\eqref{eq:R10}) and via superposition
(Eq.~\eqref{eq:DynamicRate}) will be very similar: the normalized
channel estimate $\overline{\bH}_\dyn$ in both cases has uncorrelated
entries with zero mean and unit variance.\footnote{The dynamic channel
  estimates in the orthogonal and superposition transmissions have the
  same mean and variance but are not
  identically distributed, because in the orthogonal case,
  $\overline{\bH}_\dyn$ is an estimate of $\bH_\dyn$, a Gaussian matrix,
  while in the superposition case it is an estimate of $\bH_\dyn\bX_\stat$,
  the product of two Gaussian matrices. Therefore the expectations in
  Eq.~\eqref{eq:R10} and~\eqref{eq:DynamicRate} may produce
  slightly different results.}
Thus the product superposition achieves the static user's rate ``for
free'' in the sense that the rate for the dynamic user is
approximately the same as in the single-user scenario. In the
following, we discuss this phenomenon at low and high SNR.

\subsubsection{Low-SNR Regime}
We have $\SNRaEff, \SNRbEff\ll
1$. Let the eigenvalues of $\overline{\bH}_\dyn \overline{\bH}_\dyn^{H}$ be
denoted $\bar{\lambda}_{\dyn i}^2$, $i=1,\ldots,
\Nn$. Using~\eqref{eq:DynamicRate} and a Taylor expansion of the log
function at low SNR, the achievable rate for the dynamic user is
approximately:
\begin{align}
 R_\dyn & \approx  (1-\frac{\Nn}{T_\dyn})\frac{\SNRaEff}{\Nn}  \,\Expt\big[
   \sum_{i=1}^{\Nn} \bar{\lambda}_{\dyn i}^2 \big]  \\ 
& = (1-\frac{\Nn}{T_\dyn})\frac{\SNRaEff}{\Nn}  \,\Tr \big(\Expt[\overline{\bH}_\dyn \overline{\bH}_\dyn^{H}]\big) \\
& =  (1-\frac{\Nn}{T_\dyn})\Nn \,\SNRaEff. \label{eq:DynamicLowSNR}
\end{align}
where higher-order Taylor terms have been ignored. Similarly,
from~\eqref{eq:R10}, the baseline method achieves the rate
\begin{align}
 (1-\frac{\Nn}{T_\dyn})\Nn \, \SNRaEff.
\end{align}
Thus, the dynamic user attains the same rate as it would in the
absence of the other user and its interference, i.e., a single-user
rate.  At low SNR, one cannot exceed this performance.

The rate available to the static user at low-SNR is obtained
via~\eqref{eq:StaticRate}, as follows:
\begin{align}
R_\stat & \approx \frac{\SNRbEff}{T_\dyn}  \, \Tr\big(\Expt [\bH_\stat \bH_\stat^{H}]\big) \\ 
    & = \frac{\Nc^2\,\SNRbEff}{T_\dyn}. \label{eq:StaticLowSNR}
\end{align}

\subsubsection{High-SNR Regime}
We have $\SNRaEff, \SNRbEff\gg
1$, therefore from~\eqref{eq:DynamicRate} the achievable rate for the dynamic
user is
\begin{align}
R_\dyn &\approx (1-\frac{\Nn}{T_\dyn})\bigg(\Nn\log
\frac{\SNRaEff}{\Nn}  + \Expt \big[\sum_{i=1}^{\Nn}  \log
  \bar{\lambda}_{\dyn i}^2 \big]\bigg).
\end{align}
where the approximation follows from the dominance of the channel gain
term in the $\log\det$ capacity formula. The dynamic user attains
$\Nn(1-\Nn/T_\dyn)$ degrees of freedom, which is the maximum DoF even
in the absence of the static user~\cite{Zheng2002}. Superimposing
$\bX_\stat$ only affects the distribution of eigenvalues
$\bar{\lambda}_{\dyn i}^2$, whose impact is negligible at high-SNR.

For the static user, let the eigenvalues of $\bH_\stat \bH_\stat^{H}$ be denoted
$\lambda_{\stat i}^2$, $i=1,\ldots, \Nc$. From~\eqref{eq:StaticRate}, we
have
\begin{align}
R_\stat & \approx \frac{\Nn}{T_\dyn} \bigg( \Nc\log \frac{\SNRbEff}{\Nc} +
\Expt \big[\sum_{i=1}^{\Nc} \log \lambda_{\stat i}^2 \big]\bigg),
\end{align}
which implies that the static user achieves $\Nn\Nc/T_\dyn$ degrees of
freedom. Thus, the pilot-based product superposition achieves the
DoF obtained in~\cite{Li-TIT2012-2} for $\Nn \le \Nc$, and also for
$\Nn > \Nc$.

\subsection{Power Allocation}

The effective SNRs of the dynamic and static users depend on
$c_{\tau}$ and $c_{\delta}$. We focus on $c_{\tau}$ and $c_{\delta}$ that
maximize $R_\dyn$ (equivalently $\SNRaEff$) in a manner similar
to~\cite{Hassibi2003}. From~\eqref{eq:deltaH1} and~\eqref{eq:Peff1},
\begin{align}
\SNRaEff & = \frac{c_{\tau}c_{\delta}\Nn\Nc^2}{1+c_{\tau}\Nc +
  c_{\delta}\Nn\Nc}. \label{eq:Peff11}
\end{align}
From~\eqref{eq:powerconstraint2}, we have $c_{\tau} = \rho
T_\dyn/(\Nn\Nc) - c_{\delta}(T_\dyn-\Nn)$. Substitute $c_{\tau}$
into~\eqref{eq:Peff11}:
\begin{align}
\SNRaEff & = \frac{\Nn\Nc(T_\dyn - \Nn)}{T_\dyn-2\Nn} \cdot \frac{c_{\delta} (a - c_{\delta})}{ -c_{\delta} + b},
\end{align}
where 
\begin{align}
a & =\frac{\rho T_\dyn}{\Nn\Nc(T_\dyn-\Nn)}, \\
b & =\frac{\Nn+\rho T_\dyn}{\Nn\Nc(T_\dyn-2\Nn)}.
\end{align}
Noting that $0\le c_{\delta} \le a $, we obtain the value of $c_{\delta}$ that
maximizes $R_\dyn$:
\begin{equation}
c_{\delta}^* = b-\sqrt{b^2 - ab}, \label{eq:cd}
\end{equation}
which corresponds to 
\begin{align}
\SNRaEff^* & = \frac{\Nn\Nc(T_\dyn-\Nn)}{T_\dyn-2\Nn}
\big(2b-a-2\sqrt{b^2-ab}\big), \\
\SNRbEff^* & = \frac{\rho T_\dyn}{\Nn} - \Nc(T_\dyn-\Nn) (b-\sqrt{b^2 - ab}).
\end{align}

In the low-SNR regime ($\rho \ll 1$), we have $a\ll b$, where $b\approx \frac{\Nn}{\Nn\Nc(T_\dyn-2\Nn)}$, and use Taylor expansion:
\begin{equation}
\sqrt{b^2 - ab} \approx b \big(1 - \frac{a}{2b} - \frac{a^2}{8b^2}\big). \nonumber
\end{equation} 
We obtain
\begin{align}
\SNRaEff^* & \approx \frac{\rho^2 T_\dyn^2}{4\Nn(T_\dyn-\Nn)}\\
\SNRbEff^* & \approx \frac{\rho T_\dyn}{2\Nn}.
\end{align}
This indicates that the static user has a much larger effective SNR,
i.e., $\SNRaEff^* = o (\SNRbEff^*)$. In this case,
from~\eqref{eq:DynamicLowSNR} and~\eqref{eq:StaticLowSNR}, the
achievable rate is
\begin{align}
R_\dyn& \ge \frac{T_\dyn}{4}\rho^2, \\
R_\stat&\approx \frac{\Nc}{2} \rho.
\end{align}

In the high-SNR regime where $\rho \gg 1$ we have
\begin{align}
\SNRaEff^* & \approx \frac{\rho \, T_\dyn}{(\sqrt{T_\dyn-\Nn}-\sqrt{\Nn})^2},\\
\SNRbEff^* & \approx \frac{\rho T_\dyn(\sqrt{T_\dyn/\Nn-1}-1)}{T_\dyn-2\Nn}.
\end{align}
Both static and dynamic users attain SNR that increases linearly with
$\rho$. When $T_\dyn\gg \Nn$, for the static user, $\SNRbEff^*
\approx \rho \sqrt{T_\dyn/\Nn} \gg \SNRaEff^*$. For the dynamic
user, we have $\SNRaEff^* \approx \rho$, which is the same
SNR as if the dynamic user had perfect CSI; this is not surprising
since the power used for training is negligible when the channel is
very steady.

\begin{remark}
\label{remark:power}
In the MIMO broadcast channel, conventional transmission schemes
essentially divide the power between users. In the proposed product
superposition the transmit power works for both users simultaneously
instead of being divided between them. The training power used for the
dynamic user also carries the static user's data. In this way,
significant gains over TDMA is achieved, which is contrary to the
conventional methods that at low-SNR produce little or no gain
relative to TDMA.
\end{remark}

\begin{remark}
In~\cite{Li-TIT2012-2}, the product superposition was shown to attain
the following DoF region when $\Nn \le \Nc$, i.e., achieving the coherent outer bound~\cite{Huang2009}:
\begin{equation}
\frac{d_\dyn}{\Nn} + \frac{d_\stat}{\Nc} \le 1, \quad {d_\dyn} \le {N_\dyn}(1-\frac{N_\dyn}{T_\dyn})\nonumber
\end{equation}
where $d_\dyn$ and $d_\stat$ are the DoF of the dynamic and static user, respectively.
Note that the developments in this section make no assumption about the relative number 
of antennas at the dynamic and static receivers. One can verify that Equations~\eqref{eq:DynamicRate} 
and~\eqref{eq:StaticRate} meet the  bounds shown above for both $\Nn \le \Nc$ and $\Nn > \Nc$. 
Therefore, the achievable DoF of the product superposition is now established for all dynamic/static user antenna
configurations.
\end{remark}

\section{Improving Rates by Interference Cancellation}
\label{sec:PBPS-1}

So far no interference cancellation was performed, therefore the users
did not need to decode each other's signal. However, this had the
effect that the static user utilizes only the portion of transmit
power corresponding to the dynamic user's pilot, and not the portion
corresponding to the dynamic user's data. In this section we explore
the possibility of the static user decoding the signal of the dynamic
user.\footnote{It is not necessary for the {\em dynamic} user to decode
  the other user's signal, even if it were possible, because we have
  shown the existence of static user does not significantly affect the
  capacity to the dynamic user.} To facilitate this, we concentrate on
the case $\Nc\ge\Nn$.
The received signal at the static user is
\begin{equation}
\bY_\stat = \bH_\stat\bX_\stat [\sqrt{c_{\tau}}\,\bX_{\tau}\; \sqrt{c_{\delta}}\,\bX_{\delta}] +
\bW_\stat
\end{equation}
where $\bY_\stat \in\mathcal{C}^{\Nc\times T_\dyn}$. The static user
first estimates the product
$\bH_\stat\bX_\stat\in\mathcal{C}^{\Nc\times\Nn}$ by using the pilot
$\bX_{\tau}$ sent during the first $\Nn$ time-slots, and then it
decodes $\bX_{\delta}$. Now $\bX_\dyn$ is known, therefore the entire
observed signal at the static user can be used to decode its
message. If $\bX_{\delta}$ is decoded successfully, the static user can use the power used 
by the dynamic user data, in addition to the power used by the dynamic user pilot. Intuitively, harvesting 
additional power would improve the static user's rate relative to Section~\ref{sec:PBPS}.   

Assuming the codeword used by the dynamic user is
sufficiently long, so that the static user also experiences many
channel realizations over the dynamic user codewords. 
The rate gain produced by the interference decoding is characterized by the following theorem.

\begin{theorem}
\label{thm:rate2}
Assuming $\Nc\ge\Nn$ and sufficiently long codeword of the dynamic user, with interference decoding and cancellation, the pilot-based product
superposition achieves the following rate for the static user 
\begin{equation}
R_\stat = \frac{\Nn}{T_\dyn} \Expt\bigg[ \log \det \bigg(\bI_{\Nc} +
  \frac{\SNRbEff}{\Nc} \bH_\stat\bH_\stat^{H}\bigg)\bigg], \label{eq:StaticDecodeRate}
\end{equation}
where the effective SNR is
\begin{equation}
\SNRbEff = \frac{\Nc}{\Expt[\lambda_i^{-2}]}
\end{equation}
with $\lambda_i^{2}$ being any of the unordered eigenvalues of
$\bX_\dyn\bX_\dyn^{H}$.
\end{theorem}
\begin{proof}
See Appendix~\ref{apdx:thm2}.
\end{proof}

Compared with Theorem~\ref{thm:rate1}, the SNR for the static user is
improved by using the entire $\bX_\dyn$. To see this, we decompose
$\bX_{\delta} = \bU_{\delta}\, \diag(\gamma_1,\cdots,\gamma_{\Nn}) \,
\bV_{\delta}^{H}$, and obtain
\begin{align}
\bX_\dyn\bX_\dyn^{H} & = c_{\tau} \bI_{\Nn} + c_{\delta} \,\bU_{\delta}\,
\diag(\gamma_1^2,\cdots,\gamma_{\Nn}^2)\, \bU_{\delta}^{H} \\
& = \bU_{\delta}\,\diag( c_{\tau} + c_{\delta} \gamma_1^2,\cdots,
c_{\tau} + c_{\delta}\gamma_{\Nn}^2) \,\bU_{\delta}^H.
\end{align}
Therefore, $\lambda_i^2 = c_{\tau} + c_{\delta}\,\gamma_i^2$, for $i= 1,\ldots,\Nn$, and
\begin{align}
\SNRbEff & = \frac{\Nc}{\Expt[( c_{\tau} +
    c_{\delta}\,\gamma_1^2)^{-1}]} . 
\label{eq:StaticSNR} 
\end{align}
which is greater than the effective power available to the previous
scheme (compare with Eq.~\eqref{Eq:rho2}). So knowing the dynamic
user's data always produces a power gain.

\section{Numerical Results}

Unless specified otherwise, a power allocation is assumed ($c_{\tau}$
and $c_{\delta}$) that maximizes the rate for the dynamic user.

Figure~\ref{fig:RatePBPS} illustrates the rate for dynamic and static
users in the pilot-based product superposition, as shown in
Theorem~\ref{thm:rate1}. We consider $\Nn=2$, $\Nc=M=4$ and
$T_\dyn=5$. Numerical results correspond to the point on the rate
region where the rate of the dynamic user is optimized. This is done
to capture the corner point of the DoF region for the new scheme, and
to highlight the most significant differences between the new scheme
and the baseline scheme. At this operating point, in addition to
near-optimal rate for the dynamic user, the proposed method provides
significant rate for the static user. The degradation of the rate of
the dynamic user, compared with the baseline scheme, is negligible in
the low-SNR regime, and in the high-SNR regime the rate of the dynamic
user has the optimal degrees of freedom (SNR slope). Thus the proposed
method achieves the static user's rate almost ``for free'' in terms of
the penalty to the dynamic user.
\begin{figure}
\centering
\includegraphics[width=3.7in]{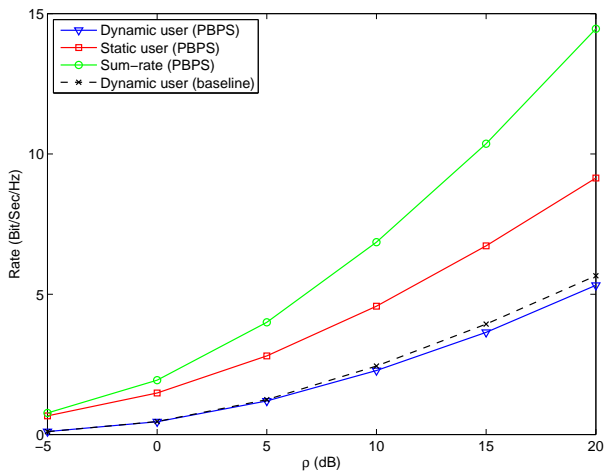}
\caption{Rate achieved by the pilot-based product superposition
  (PBPS): $\Nn=2$, $\Nc=M=4$ and $T_\dyn=5$.}
\label{fig:RatePBPS}
\end{figure}

Figure~\ref{fig:AntennaPBPS} shows the impact of the available antenna
of the static user. Here, $\rho=10$ dB, $\Nn=2$, $M=\Nc$ and
$T_\dyn=5$. The static user's rate (thus the sum-rate) increases
linearly with $\Nc$, because the degrees of freedom is
$\Nn\Nc/T_\dyn$, as indicated by Theorem~\ref{thm:rate1}. The gap of
the dynamic user's rate under the proposed method and the baseline
method vanishes as $\Nc$ increases. Intuitively, the rate difference
is because of the Jensen's loss: in the proposed method the equivalent
channel is the product matrix $\bH_\dyn\bX_\stat$ and is ``more
spread'' than the channel in the baseline method. As $\Nc$ increases,
by law of large numbers the columns of $\bX_\stat$ will become
orthonormal with probability one
($\bX_\stat\bX_\stat^{H}/\Nc\rightarrow \bI_{\Nn}$) and thus will have
a smaller impact on the distribution of $\bH_\dyn$.

\begin{figure}
\centering
\includegraphics[width=3.7in]{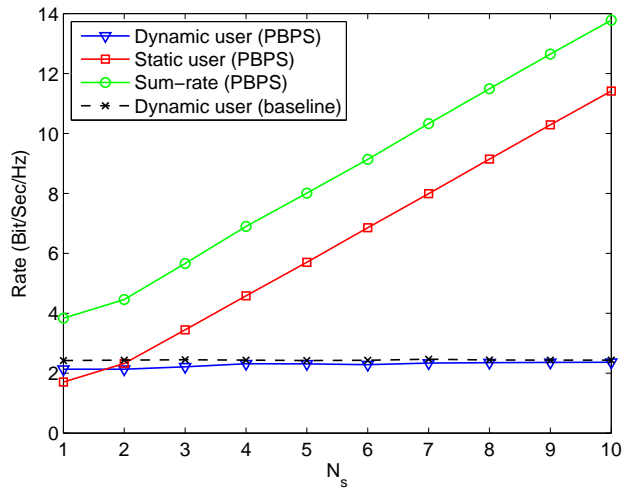}
\caption{Impact of the number of receive antennas of the static user: $\rho=10$ dB,
  $\Nn=2$, $M=\Nc$ and $T_\dyn=5$.}
\label{fig:AntennaPBPS}
\end{figure}

Figure~\ref{fig:T} demonstrates the impact of the 
coherence time of the dynamic user. Here, $\rho=10$ dB, $\Nn=2$, and
$\Nc=M=4$. As $T_\dyn$ increases, the rate for the dynamic user improves,
since the portion of time-slots (overhead) used for training is
reduced. In contrast, the rate for the static user decreases with
$T_\dyn$, because the static user transmits new signal matrix over $T_\dyn$
period. Intuitively, as $T_\dyn$ increases, the dynamic user's channel
becomes ``more static'', and therefore, the opportunity to explore its
``insensitivity'' to the channel is reduced. 

\begin{figure}
\centering
\includegraphics[width=3.7in]{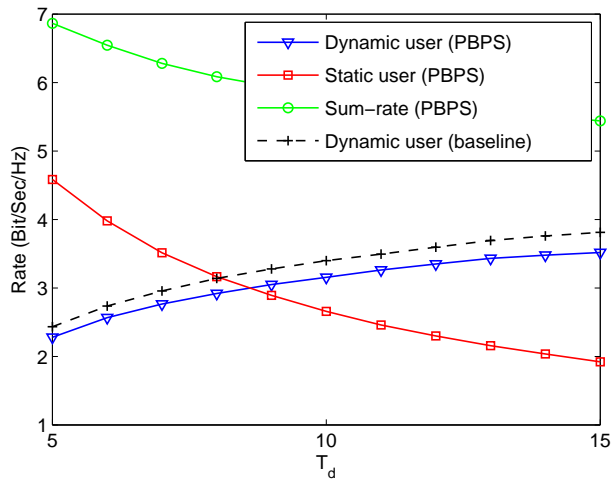}
\caption{Impact of channel coherence time: $\rho=10$ dB,
  $\Nn=2$, and $\Nc=M=4$.}
\label{fig:T}
\end{figure}

Finally, in Figure~\ref{fig:IC}, we show the gain of interference
decoding in the pilot-based product superposition, where $\Nn=2$,
$\Nc=M$ and $T_\dyn=5$. By decoding the dynamic signal , the static rate is improved
around $10\%$: the static user can now harvest the power carried not only by 
the dynamic user's pilot (the case without interference decoding) but also the dynamic user's data.
This power gain does not increase the degrees of freedom of the static user, so the slope 
of the rate under two schemes are the same.

\begin{figure}
\centering
\includegraphics[width=3.7in]{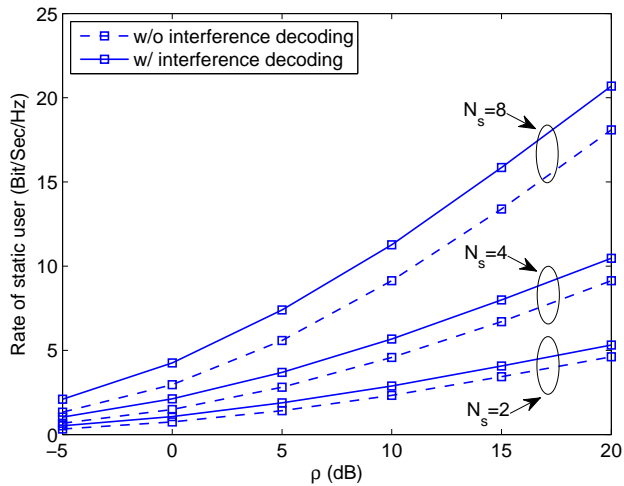}
\caption{Static user's rate with interference decoding: $\Nn=2$ and $T_\dyn=5$.}
\label{fig:IC}
\end{figure}

\section{Discussions, Extensions, and Conclusion}

In this paper, we propose and analyze a pilot-based signaling
that significantly improves the rate performance of the MIMO broadcast
channel with varying CSIR. The proposed method sends a product of two
signal matrices for the static and dynamic user, respectively, and
each user decodes its own signal in a conventional manner.  
For the entire SNR range, the static user attains
considerable rate almost without degrading the rate for the dynamic
user. The static user's rate is further improved by allowing the
static user to cancel the dynamic user's signal. 


\begin{remark}
It is possible to extend the results of this paper to more than two
receivers. The essence of the product superposition is to allow additional transmission 
for a static user when transmitting to a dynamic user. 
In case of more than two users, the static (dynamic) users can be grouped together.
At each point in time, the transmitter uses product superposition to broadcast to one selected 
user from the static group and another user from the dynamic group.

\end{remark}
\begin{remark}
\label{remark:nonergodic}
Note that throughout this paper, both users are assumed to be in an
ergodic mode of operation, i.e., the codewords are sufficiently long
to allow coding arguments to apply. Simple extensions to this setup
are easily obtained. For example, if the static user's coherence time
is very long, one may adapt the transmission rate of the static user
to its channel but allow the dynamic user to remain in an ergodic
mode. Most expressions in this paper remain the same, except that for
the rates and powers of the static user, expected values will be
replaced with constant values.
\end{remark}
\begin{remark}
As long as both users are in the ergodic mode, and the static user has
more antennas than the dynamic user, it will be able to decode and
cancel the interference caused by the dynamic user's signal. If we are
in a mode where the static user's rate is adapted to the channel (as
mentioned in Remark~\ref{remark:nonergodic} above) and the dynamic
user is in ergodic mode, then the static user may not always be able
to decode the dynamic user's data because it cannot observe enough
channel realizations to allow coding arguments to apply. In this case,
sometimes the static user may experience an ``outage'' with respect to
decoding the dynamic user's data. In this case, it can default to the
oblivious method discussed in the early part of this paper and decode
its own signal without peeling off the other user's data. The full
exploration of such extensions is the subject of future research.
\end{remark}

\begin{remark}
In each of the methods mentioned earlier in this paper, the static
user operates under an equivalent single-user channel, by inverting
either the pilot component or all components of the dynamic user's
signal. Thus, any benefits available in single-user
MIMO systems can also be available to the static user, including the benefits
arising from CSIT.  For example, water-filling can be applied to
allocate power across multiple eigen-modes of the static
user. However, this will change the effective channel seen by the
dynamic user, thus complicating the analysis. The full analysis of
this scenario is the subject of future research.
\end{remark}

\section*{Acknowledgment}

The authors gratefully acknowledge the valuable help of Mr. Mohamed Fadel.

\appendices 

\section{Proof Of Theorem~\ref{thm:rate1}}
\label{apdx:thm1}
\subsection{Rate of the Static User}
During the first $\Nn$ time-slots, the static user receives
\begin{equation}
\bY_{\stat 1 } = \sqrt{c_{\tau} }\;\bH_\stat \bX_\stat \bX_{\tau} +
\bW_{\stat 1}.
\end{equation}
Because the static user knows $\bX_{\tau}$, it removes the impact of
$\bX_{\tau}$ from $\bY_{2\tau}$:
\begin{align}
\bY_{\stat 1}^{\prime} & = \bY_{\stat 1} \bX_{\tau}^{H} \\
& = \sqrt{c_{\tau}}\;\bH_\stat \bX_\stat  + \bW_{\stat 1}^{\prime}
\end{align}
where $\bY_{\stat 1} \in \mathcal{C}^{\Nc\times \Nn}$ and
$\bW_{\stat 1}^{\prime}$ is the equivalent noise whose entries remain
i.i.d. $\mathcal{CN}(0,1)$. Therefore, the channel seen by the static
user becomes a point-to-point MIMO channel. Let $\by_{\stat i}^{\prime}$
and $\bx_{\stat i}$ be the column $i$ of $\bY_{\stat 1}^{\prime}$ and
$\bX_\stat$, respectively. The mutual information
\begin{align}
I(\bY_{\stat 1};\bX_\stat) & = \sum_{i=1}^{\Nn}
I(\by_{\stat i}^{\prime};\bx_{\stat i}) \\ & = \Nn \log\det\bigg(\bI_{\Nc} +
c_{\tau}\, \bH_\stat \bH_\stat^{H}\bigg),
\end{align} 
which implies that the effective SNR for the static user is
\begin{equation}
\SNRbEff = c_{\tau}.
\end{equation}

In the following $T_\dyn-\Nn$ time-slots, the static user disregards
the received signal. The average rate achieved by the
static user is
\begin{equation}
R_\stat = \frac{\Nn}{T_\dyn} \,\Expt \bigg[\log\det\bigg(\bI_{\Nc} +
  \SNRbEff \, \bH_\stat \bH_\stat^{H}\bigg)\bigg], 
\end{equation}
where the expectation is over the channel realizations of $\bH_\stat$.

\subsection{Rate of the Dynamic User}
The dynamic user first estimates the equivalent channel and then
decodes its data. 
During the first $\Nn$ time-slots, the dynamic user receives the pilot
signal
\begin{align}
\bY_{\tau} & = \sqrt{c_{\tau} }\;\bH_\dyn \bX_\stat \bX_{\tau} +
\bW_{\tau} \\
& = \sqrt{c_{\tau} \Nc} \;\Heqd \bX_{\tau} +
\bW_{\tau},
\end{align}
where $\Heqd \in\mathcal{C}^{\Nn\times \Nn}$ is the equivalent channel
of the dynamic user
\begin{equation}
\Heqd \stackrel {\Delta}{=} \frac{1}{\sqrt{\Nc}}\,\bH_\dyn\bX_\stat
\end{equation}
Let $\tilde{h}_{ij}=[\Heqd]_{ij}$, then we have
$\Expt[\tilde{h}_{ij}]=0$ and
\begin{equation}
\Expt[\tilde{h}_{ij}\,\tilde{h}_{pq}^{H}] = 
\begin{cases}
1,\ \ \text{if} \ (i,j)=(p,q)\\
0, \ \ \text{else}
\end{cases},
\end{equation}
i.e., the entries of $\Heqd$ are uncorrelated and have zero-mean and
unit variance.

The dynamic user estimates $\Heqd$ by the MMSE. Let
\begin{align}
C_{YY}= ( 1+c_{\tau}\Nc) \bI_{\Nn},\quad C_{YH} =  \sqrt{c_{\tau} \Nc} \; \bX_{\tau}^{H}, 
\end{align}
we have
\begin{align}
\widehat{\bH}_\dyn & = \bY_{\tau} C_{YY}^{-1} C_{YH} \\ 
 & = \frac{ \sqrt{c_{\tau} \Nc} }{1+c_{\tau}\Nc} \bigg(\sqrt{c_{\tau}
  \Nc} \;\Heqd + \bW_{\tau} \bX_{\tau}^{H}\bigg)
\end{align}
Because $\bW_{\tau}$ has i.i.d. $\mathcal{CN}(0,1)$ entries, the
noise matrix $\bW_{\tau} \bX_{\tau}^{H}$ also has
i.i.d. $\mathcal{CN}(0,1)$ entries. 
Define $\hat{h}_{1ij}=[\widehat{\bH}_\dyn]_{ij}$. Then, we have
$\Expt[\hat{h}_{1ij}]=0$ and
\begin{equation}
\Expt[\hat{h}_{ij}\hat{h}_{pq}^{H}] = 
\begin{cases}
\alpha^2 ,\ \ \text{if} \ (i,j)=(p,q)\\
0, \ \ \ \ \ \text{else}
\end{cases},
\end{equation}
where
\begin{equation}
\alpha^2 \defeq  \frac{ c_{\tau} \Nc }{1+c_{\tau}\Nc}. \label{eq:deltaH1}
\end{equation}
In other words, the estimate of the equivalent channel has
uncorrelated elements with zero-mean and variance
$\alpha^2$.

During the remaining $T_\dyn-\Nn$ time-slots, the dynamic user regards
the channel estimate $\widehat{\bH}_\dyn$ as the true channel and decodes
the data signal. At the time-slot $i$, $\Nn<i\le T_\dyn$, the dynamic
user receives
\begin{align}
\by_{\dyn i} = \sqrt{c_{\delta} \Nc} \; \widehat{\bH}_\dyn \bx_{\dyn i} +
\underset{\bw_{\dyn i}^{\prime}}{\underbrace{ \sqrt{c_{\delta} \Nc} \;
    \widetilde{\bH}_{e} \bx_{\dyn i}
+ \bw_{\dyn i}}},
\end{align}
where $\widetilde{\bH}_{e} = \Heqd - \widehat{\bH}_\dyn$ is the 
estimation error for $\Heqd$, and $\bw_{\dyn i}^{\prime}$ is the equivalent noise that
has zero mean and autocorrelation
\begin{align}
\bR_{w^{\prime}_\dyn} & = c_{\delta} \Nc
\;\Expt\big[\widetilde{\bH}_{e}\widetilde{\bH}_{e}^{H}\big] +
\bI_{\Nn} \\
& = \big(1+\frac{c_{\delta}\Nn\Nc }{1+c_{\tau}\Nc}\big)\bI_{\Nn}. 
\end{align}
 The equivalent noise $\bw_{ \dyn i}'$ is uncorrelated with the signal
 $\bx_{\dyn i}$, because $\Expt[ \tilde{\bH}_{e} \bx_{\dyn i}
   \bx_{\dyn i}^H]
 =\Expt[ \tilde{\bH}_{e} ]  \Expt[ \bx_{\dyn i} \bx_{\dyn i}^H]  =0
 $. Therefore, from~\cite[Thm.1]{Hassibi2003}, the mutual information
 is lower bounded by:
\begin{align}
I(\by_{\dyn i};\bx_{\dyn i}|\widehat{\bH}_\dyn) & \ge \log\det\bigg(\bI_{\Nn} +
\frac{c_{\delta}\Nc\;\widehat{\bH}_\dyn\widehat{\bH}_\dyn^{H}}{1+c_{\delta}\Nn\Nc/(1+c_{\tau}\Nc)}\bigg)
\\
& = \log\det\bigg(\bI_{\Nn} +
\frac{c_{\delta}
  \alpha^2\Nc\;\overline{\bH}_\dyn\overline{\bH}_\dyn^{H}}{1+c_{\delta}\Nn\Nc/(1+c_{\tau}\Nc)}\bigg), \label{eq:MIR1}
\end{align}
where $\overline{\bH}_\dyn$ is the normalized channel whose elements have
unit variance
\begin{equation}
\overline{\bH}_\dyn = \frac{1}{\alpha} \widehat{\bH}_\dyn.
\end{equation}
From~\eqref{eq:MIR1}, the effective SNR for the dynamic user can be defined as
\begin{equation}
\SNRaEff =
\frac{c_{\delta}\alpha^2\Nn\Nc}{1+c_{\delta}\Nn\Nc/(1+c_{\tau}\Nc)}. \label{eq:Peff1}
\end{equation}
The average rate that the dynamic user achieves is
\begin{equation}
R_\dyn \ge (1-\frac{\Nn}{T_\dyn}) \Expt \big[\log\det(\bI_{\Nn} +
  \frac{\SNRaEff}{\Nn} \overline{\bH}_\dyn\overline{\bH}_\dyn^{H})\big], 
\end{equation}
where the expectation is over the dynamic user's channel realizations.

\section{Proof of Theorem~\ref{thm:rate2}}
\label{apdx:thm2}

We first show that if the codeword used by the dynamic user is
sufficiently long, the static user always decodes the dynamic user's
signal.

Similar to the dynamic user, the equivalent channel of the static user
$\Heqs \stackrel {\Delta}{=} \bH_\stat\bX_\stat/\sqrt{\Nc}$ can be estimated as 
$\widehat{\bH}_\stat\in \mathcal{C}^{\Nc\times\Nn}$ by using the pilot $\bX_{\tau}$. During time-slots $i =\Nn+1,\ldots, T_\dyn$,
the static user receives:
\begin{align}
\by_{\stat i} = \sqrt{c_{\delta} \Nc} \; \widehat{\bH}_\stat \bx_{\dyn
  i} +
\underset{\bw_{\stat i}^{\prime}}{\underbrace{ \sqrt{c_{\delta} \Nc}
    \; \widetilde{\bH}_{e} \bx_{\dyn i}
+ \bw_{\stat i}}},
\end{align}
where $\bx_{\dyn i} \in \mathcal{C}^{\Nn\times 1}$ is the $i$-th column of $X_\dyn$.
The mutual information
\begin{align}
I(\by_{\stat i};\bx_{\dyn i}|\widehat{\bH}_\stat) & \ge \log\det\bigg(\bI_{\Nc} +
  \frac{c_{\delta}\Nc\; \widehat{\bH}_\stat
    \widehat{\bH}_\stat^{H}}{1+c_{\delta}\Nn\Nc/(1+c_{\tau}\Nc)} 
\bigg) \\
& = \log\det\big(\bI_{\Nn} + \frac{\SNRaEff}{\Nn} \overline{\bH}_\stat\overline{\bH}_\stat^{H}\big),
\end{align}
where $\overline{\bH}_\stat = \frac{1}{\alpha}\widehat{\bH}_\stat$ is
the normalized channel estimate and $\rho_\dyn$ was given
in~\eqref{eq:SNRdynamic}. For the static user, the effective SNR for
decoding the dynamic signal is identical to that of the dynamic user.

The static user also experiences many
channel realizations over the dynamic user codewords. Write $\overline{\bH}_\stat =
[\overline{\bH}_{\stat 1} ; \overline{\bH}_{\stat 2}]$, where
$\overline{\bH}_{\stat 1} \in\mathcal{C}^{\Nn\times \Nn}$ and
$\overline{\bH}_{\stat 2} \in\mathcal{C}^{(\Nc-\Nn)\times \Nn}$. Then,
\begin{align}
\Expt \big[I(\by_{\stat i};& \bx_{\dyn i} |\widehat{\bH}_\stat)\big] \nonumber \\
 &\ge
\Expt\bigg[\log\det\bigg(\bI_{\Nn} + \SNRaEff
\big(\overline{\bH}_{\stat 1}\overline{\bH}_{\stat
  1}^{H}+\overline{\bH}_{\stat 2}\overline{\bH}_{\stat 2}^{H}
\big)\bigg)\bigg]\\
& \ge \Expt\bigg[\log\det\bigg(\bI_{\Nn} +
  \SNRaEff\overline{\bH}_{\stat 1}\overline{\bH}_{\stat
    1}^{H}\bigg)\bigg],\label{eq:bound1}\\
&= \Expt\bigg[\log\det\bigg(\bI_{\Nn} +
  \SNRaEff{\bH}_{\dyn}{\bH}_{\dyn}^{H}\bigg)\bigg],\label{eq:equiv}\\
&= R_\dyn
\end{align}
where~\eqref{eq:bound1} uses $\log\det(\bA+\bB)\ge \log\det \bA$ for
positive definite matrices $\bA, \bB$, and~\eqref{eq:equiv} uses the
fact that $\overline{\bH}_{\stat 1}$ has the same distribution as
$\overline{\bH}_\dyn$. Therefore the static user can decode the
dynamic user's signal, and from here on we assume the static user has
access to the dynamic user signal.


We now use the singular value decomposition of the dynamic signal
$\bX_\dyn = \bU_\dyn\mathbf{\Sigma}_\dyn\bV_\dyn^{H}$, where
$\bU_\dyn\in \mathcal{C}^{\Nn\times \Nn}$, $\bV_\dyn\in
\mathcal{C}^{T_\dyn \times \Nn}$ are unitary matrices, and
$\mathbf{\Sigma}_\dyn=\diag(\lambda_1,\cdots,\lambda_{\Nn})$. Then, we
have
\begin{align}
\bY_\stat^{\prime} & = \bY_\stat \bV_\dyn \mathbf{\Sigma}_\dyn^{-1} \\
& = \bH_\stat \bX_\stat \bU_\dyn + \bW_\stat\bV_\dyn\mathbf{\Sigma}_\dyn^{-1} \\
& \stackrel{\Delta}{=}  \bH_\stat \bX_\stat^{\prime} + \bW_\stat^{\prime}
\mathbf{\Sigma}_\dyn^{-1}, \label{eq:StaticDecodePost}
\end{align}
where $\bX_\stat^{\prime} = \bX_\stat \bU_\dyn,  \bW_\stat^{\prime}= \bW_\stat\bV_\dyn$.
Because $\bU_\dyn$, $\bV_\dyn$ are unitary, the entries of $\bX_\stat^{\prime},
\bW_\stat^{\prime} \in \mathcal{C}^{\Nc\times \Nn}$ remain
i.i.d. $\mathcal{CN}(0,1)$. Define $\by_\stat^{\prime} = \vec(\bY_\stat^{\prime})$,
$\bx_\stat^{\prime} = \vec( \bX_\stat^{\prime})$, $\bH_\stat^{\prime} =
\bI_{\Nn}\otimes \bH_\stat$ and 
\begin{align}
\bw_\stat^{\prime} & = \vec(\bW_\stat^{\prime}
\mathbf{\Sigma}_\dyn^{-1}) = \begin{bmatrix} \frac{1}{\lambda_1}
  \bw_{\stat 1}^{\prime} \\ \vdots
  \\  \frac{1}{\lambda_{\Nn}}  \bw_{\stat \Nn}^{\prime} \end{bmatrix}.
\end{align}
Then,
from~\eqref{eq:StaticDecodePost}, we write
$\by_\stat^{\prime}\in\mathcal{C}^{\Nn\Nc\times 1}$ as
\begin{equation}
\by_\stat^{\prime} = \bH_\stat^{\prime}  \bx_\stat^{\prime}  +
\bw_\stat^{\prime}.
\end{equation}
The mutual information
\begin{align}
I(\bY_\stat;\bX_\stat | \bH_\stat, \bX_\dyn) & =  I(\by_\stat^{\prime};\bx_\stat^{\prime} | \bH_\stat, \bX_\dyn)
\\
& = \log \det \bigg(\bI_{\Nn\Nc} + \bR_{w_\stat^{\prime}}^{-1} \ \bH_\stat^{\prime}\bH_\stat^{\prime\,H}\bigg),
\end{align}
where $\bR_{w_\stat^{\prime}} = \Expt[\bw_\stat^{\prime}
  \bw_\stat^{\prime\,H}]$ is the noise autocorrelation matrix that is given
by
\begin{equation}
\bR_{w_\stat^{\prime}} = \begin{bmatrix} \Expt[\lambda_1^{-2}]
  \bI_{\Nc} & & \\ & \ddots & \\ & & \Expt[\lambda_{\Nn}^{-2}]
  \bI_{\Nc} \end{bmatrix}.
\end{equation}

Therefore, the average rate attained by the static user is
\begin{align}
R_\stat  & = \frac{1}{T_\dyn} \Expt[I(\bY_\stat;\bX_\stat | \bH_\stat, \bX_\dyn)] \\
& =  \frac{1}{T_\dyn} \Expt\bigg[ \sum_{i=1}^{\Nn}\log \det \bigg(\bI_{\Nc} +
  \frac{1}{\Expt[\lambda_i^{-2}]} \bH_\stat\bH_\stat^{H}\bigg)\bigg] \\
& = \frac{\Nn}{T_\dyn} \Expt\bigg[ \log \det \bigg(\bI_{\Nc} +
  \frac{1}{\Expt[\lambda_1^{-2}]} \bH_\stat\bH_\stat^{H}\bigg)\bigg], 
\end{align}
where the last equality holds because the marginal distributions of
$\{\lambda_i\}$ are identical.

\bibliographystyle{IEEEtran} \bibliography{IEEEabrv,LiYang-revised2}

\end{document}